\documentclass[12pt]{amsart}
\usepackage[margin=1.0in]{geometry}
\usepackage{mathptmx,color,amssymb,graphicx}
\usepackage[foot]{amsaddr}
{\theoremstyle{plain}
   
   \newtheorem{theorem}{Theorem}[section]
   \newtheorem{proposition}[theorem]{Proposition}

}
{\theoremstyle{definition}

   \newtheorem{remark}[theorem]{Remark}

}

\begin{document}
\title{Computational geometry and the U.S. Supreme Court}
\author{Noah Giansiracusa$^{*}$} 
\address{$^*$Assistant Professor of Mathematics, Swarthmore College, and corresponding author: ngiansi1@swarthmore.edu} 
\author{Cameron Ricciardi$^{**}$} 
\address{$^{**}$Undergraduate Mathematics and Economics Major, Swarthmore College}

\maketitle

\begin{abstract}
We use the United States Supreme Court as an illuminative context in which to discuss three different spatial voting preference models: an instance of the widely used single-peaked preferences, and two models that are more novel in which vote outcomes have a strength in addition to a location.  We introduce each model from a formal axiomatic perspective,  briefly discuss practical motivation for each in terms of judicial behavior, prove mathematical relationships among the voting coalitions compatible with each model, and then study the two-dimensional setting by presenting computational tools for working with the models and by exploring these with judicial voting data from the Supreme Court.
\end{abstract} 
\vspace{0.1in}


\section{Introduction}

A popular view among legal scholars, journalists, and amateur court-watchers is that the nine justices sitting on the bench of the U.\ S.\ Supreme Court are driven largely by political ideology \cite{SegalSpaeth}.  To quantify this perspective, one considers each justice as having an ``ideal point'' on a one-dimensional axis\footnote{There have been several sophisticated statistical methods proposed to estimate the location of the ideals points, using data such as judicial voting records, op-ed authorship, party of presidential appointment, etc. \cite{MDS-vs-factor,JudicialCommonSpace,MartinQuinn,Peress-IdealPoint,SegalCover}.} and uses this layout in voting models to help understand the ways the justices vote, and in particular the ways the justices group themselves into a majority coalition and a minority coalition.  While the majority/minority dichotomy masks the complexity of the concurring and dissenting opinions underlying each decision, it is a simplification that many scholars---ourselves included---are willing to make, so we shall assume that each case has two outcomes, affirm or reverse, and that each justice votes for exactly one of these.

Voting models often assume ``single peaked'' preferences, which here means the two potential outcomes of a case are placed along the same one-dimensional political axis as the justices' ideal points, and if a justice is to the left of both outcomes or to the right of both then the justice votes for the closer of the two outcomes.  If a justice lies between the two outcomes then the vote depends on the precise shape of the justice's preference function; since this is difficult to estimate, a natural simplifying assumption is that all the justices use negated distance for their preference function, which means they just vote for the closest outcome to their ideal point on the political axis.  The midpoint between the two case outcomes then serves as a dividing point so that all justices to the left of it vote one way and all justices to the right vote the other way.  In particular, each majority coalition provided by this model is separated from the minority coalition by a point on the political axis.  But how does one then make sense of scrambled situations like the 5-to-4 vote in \emph{Ginzburg v.\ United States} (1965), where the minority coalition consisted of far left Douglas, far right Harlan, the right-leaning Stewart, and the centrist Black? (See Figure \ref{fig:1mds}.)

\begin{figure}
\begin{center}
\includegraphics[trim={0.8in 0.6in 0.4in 0.8in},clip,scale=1.15]{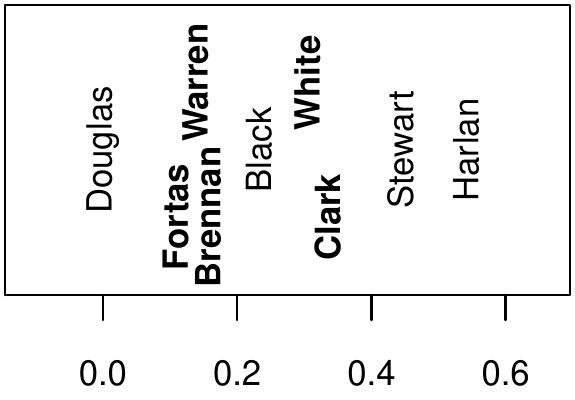}
\end{center}
\caption{A one-dimensional layout of judicial ideal points for the 1965--1966 natural court, using multidimensional scaling based on the voting similarity matrix (the details of these methods are discussed later in the paper).  The names for the majority in \emph{Ginzburg v.\ United States} (1965) are in bold.}
\label{fig:1mds}
\end{figure}

The dimensionality of the Supreme Court has been an active topic of exploration among scholars \cite{MDS-vs-factor,PNAS-2003-Sirovich,DimSupCourtEdelman,StatMechSC}, and indeed some of the majority/minority divisions that seem illogical from the political spectrum perspective naturally reveal themselves on a second dimension \cite{2ndDimension}.  However, a two-dimensional binary outcome voting model based on Euclidean distance preference functions still only allows for votes in which the majority is separated from the minority by a line in the plane, and while many cases do nicely exhibit this structure, many yet do not (including, for instance, the scrambled 5-4 vote mentioned above).  We introduce here two voting models, one which allows more geometrically flexible majority/minority divisions and one which allows the same divisions as the Euclidean distance model but arrives at these divisions in a different manner.  These models might be novel, but both authors of this paper come from a mathematical---specifically, geometric---background, so we are less familiar with the voting theory literature; our emphasis then is not on originality, per se, of these two voting models but on comparing the coalitions provided by all three models in two dimensions and presenting methods from computational geometry that allow a practitioner to work with and explore them.  

In \S\ref{sec:models} we introduce the three models from a formal mathematical perspective, provide a geometric interpretation, and briefly discuss some motivation for the models in the context of Supreme Court judicial voting behavior.  We then prove a theorem relating the coalitions---which is to say, the majority/minority divisions---allowed by each model; this is done for a Euclidean space $\mathbb{R}^d$ of arbitrary dimension $d\ge 1$.  The section concludes with an illustration of the models in two dimensions on a few case votes from the Supreme Court.  In \S\ref{sec:compgeom} we present computational methods for finding all the coalitions allowed by each model for a fixed configuration of voter ideal points, and for one of the models we also show how to estimate the locations of the two case outcomes based on the votes; this too is illustrated with Supreme Court cases.  We also mention connections to some topics in discrete/computational geometry, such as $k$-sets \cite{k-set} and higher order Voronoi diagrams \cite{HigherVoronoi}.  We conclude in \S\ref{sec:conclusion} with a few closing remarks.


\subsection*{Acknowledgements}

We thank Paul Edelman for helpful conversations on Supreme Court data and the literature analyzing it, Michael Burr for helpful conversations on computational geometry algorithms, and John Roemer for invaluable feedback and suggestions based on an earlier draft of this paper.  The first author was supported with funds from the NSF grant DMS-1802263.


\section{Three spatial voting models}\label{sec:models}

For each of the three voting models under consideration we: provide a formal mathematical description and geometric interpretation; briefly discuss a rationale for the model in terms of judicial behavior; and study the majority/minority coalitions allowed by the model in all dimensions, with illustrations in dimensions one and two coming from Supreme Court cases.

\subsection{Formal definitions and judicial motivation}

\subsubsection{General setup}
In all three of the settings discussed here we suppose each of the $N$ voters has an \emph{ideal point} in Euclidean space, let us denote these \[\mathbf{v}^{(1)}=(v_1^{(1)},\ldots,v_d^{(1)}),~\ldots,~\mathbf{v}^{(N)}=(v_1^{(N)},\ldots,v_d^{(N)})\in \mathbb{R}^d.\]  The component $v_j^{(i)}$ represents the placement of the $i$-th voter along the $j$-th axis; typically the axes represent some quality influencing the voter's outlook, such as political ideology, but such an interpretation is immaterial for the mathematical treatment itself.  

We only consider binary votes, meaning the space of possible outcomes for each vote has only two items (\emph{affirm} and \emph{reverse}, in the case of the Court).  As is common in voting theory, for each vote (or \emph{case} in our judicial setting) in question, each voter (or \emph{justice}) has an associated preference function---let us write this as \[f_i : \{\mathbf{p},\mathbf{q}\} \rightarrow \mathbb{R}\] for the $i$-th voter---and the voter always votes for the outcome with a greater preference value (we do not need to identify the set $\{\mathbf{p},\mathbf{q}\}$ with affirm and reject in any particular way as we shall treat both possible outcomes equally from a mathematical perspective).  Thus a voting model, in this context, is a way of specifying the preferences functions of all $N$ voters.  

In all three of the models we consider, the outcome space is embedded in the same Euclidean space as the voters,
\[\mathbf{p}=(p_1,\ldots,p_d),~\mathbf{q}=(q_1,\ldots,q_d) \in \mathbb{R}^d,\]
and the preference functions are derived in some way from the spatial relations among the ideal points and the outcome points; we shall refer to this as a \emph{spatial} voting model.  In the judicial context, when $d=1$ we shall view smaller values in $\mathbb{R}$ as more politically liberal and larger values as more politically conservative, so each justice has their own hypothesized location on this political axis and each of the two possible outcomes of each case before the Court also has a hypothesized location on this axis.  We next describe the preference functions for our three different spatial models.

\subsubsection{Euclidean Distance Preference (EDP)}

This is a special case of single-peaked voting preferences where we simply use the Euclidean distance formula to define the preference functions:
\[f_i(\mathbf{p}) = -\sqrt{(v^{(i)}_1 - p_1)^2 + \cdots + (v^{(i)}_d - p_d)^2},~f_i(\mathbf{q}) = -\sqrt{(v^{(i)}_1 - q_1)^2 + \cdots + (v^{(i)}_d - q_d)^2}.\]  The motivation for this model is rather straightforward: the closer a case outcome is to a justice's personal outlook, the more inclined the justice is to vote for it---this suggests the use of a metric on $\mathbb{R}^d$, and the Euclidean metric is as natural as any.

\subsubsection{Sphere of Influence (SOI)}

Here we suppose the case under consideration associates not just a Euclidean location to the possible outcomes, but also a \emph{strength}, which we shall denote by \[s : \{\mathbf{p},\mathbf{q}\} \rightarrow \mathbb{R}_{\ge 0}.\]  We then have the following preference functions:
\[f_i(\mathbf{p}) = 
\begin{cases} 
      1\text{ if }\sqrt{(v^{(i)}_1 - p_1)^2 + \cdots + (v^{(i)}_d - p_d)^2}  \le s(\mathbf{p})\\
      0\text{ otherwise},
\end{cases}
\]
\[f_i(\mathbf{q}) = 
\begin{cases} 
      1\text{ if }\sqrt{(v^{(i)}_1 - q_1)^2 + \cdots + (v^{(i)}_d - q_d)^2}  \le s(\mathbf{q})\\
      0\text{ otherwise}.
\end{cases}
\]
Thus the $i$-th voter likes an outcome if the voter's ideal point lies within a sphere of radius equal to the strength of the outcome (which we shall term an \emph{outcome sphere}) and dislikes the outcome otherwise.  This can frequently result in ties---both outcomes having equal preference for a voter---so we will only consider this model under the assumption that each voter lies within \emph{exactly} one of the two outcome spheres.  Note, however, that it does not matter whether the two outcome spheres overlap, as long as there are no voter ideal points within the intersection.

The legal interpretation of this model is that a justice may have an innate inclination toward a particular outcome (based on political ideology, for instance) but the willingness of the justice to vote for that outcome also depends on the strength of the legal argument put before the court.  If both sides have equally strong arguments then the justice votes for the nearest outcome just as in the EDP model---but the justice is more willing to switch to a politically further away outcome if that side of the case is argued more strongly.  One might consider a more general setup in which the strengths of the two sides vary among the voters (as different justices will certainly measure the strength of an argument differently) but we shall not do this in the present paper; the model already provides plenty to explore with a single strength function for all voters.

\subsubsection{Take It or Leave It (TIOLI)}

Here we take an unorthodox approach by breaking symmetry and have the preference functions determined by properties of just one of the two possible outcomes, without loss of generality $\mathbf{p}$---and like the SOI model, we assign a strength $s(\mathbf{p})$ to this outcome in addition to its Euclidean location $\mathbf{p}\subseteq \mathbb{R}^d$.  Specifically, we use the following preference functions:
\[f_i(\mathbf{p}) = 
\begin{cases} 
      1\text{ if }\sqrt{(v^{(i)}_1 - p_1)^2 + \cdots + (v^{(i)}_d - p_d)^2}  \le s(\mathbf{p})\\
      0\text{ otherwise},
\end{cases}
\]
\[f_i(\mathbf{q}) = 
\begin{cases} 
      1\text{ if }\sqrt{(v^{(i)}_1 - p_1)^2 + \cdots + (v^{(i)}_d - p_d)^2}  > s(\mathbf{p})\\
      0\text{ otherwise}.
\end{cases}
\]
Thus a voter selects outcome $\mathbf{p}$ if the voter's ideal point lies in a sphere of radius $s(\mathbf{p})$, otherwise the voter selects outcome $\mathbf{q}$.  This can be viewed in rough analogy to the concept in statistics of rejecting the null hypothesis: the voter can be thought of as defaulting to outcome $\mathbf{q}$ but is willing to switch to $\mathbf{p}$ if sufficiently compelled by it.  In the judicial context, being compelled to switch to outcome $\mathbf{p}$ means having an ideology in reasonably close proximity to this outcome---and exactly how close this proximity needs to be is measured by the strength of the argument for $\mathbf{p}$.  For instance, a justice may generally opt to affirm cases brought to the court but is willing to reverse them if the political orientation of doing so is sufficiently compatible with the justice's innate outlook and the legal argument for reversing is also sufficiently convincing.

\subsection{Coalitions in one dimension}

Since we are focusing on votes with two outcomes, each vote divides the voters into two disjoint subsets: those who vote for $\mathbf{p}$, and those who vote for $\mathbf{q}$.  Let us refer to these two subsets as \emph{coalitions}.  For the Supreme Court, this means there is a \emph{majority coalition} and a \emph{minority coalition}, since the fully staffed Court has an odd number of justices so one of the two outcomes always receives strictly more votes than the other.  In the remainder of this section we explore the possible coalitions allowed by our three voting models; here by an \emph{allowed coalition} we mean that for a fixed configuration of voter ideal points, there exists locations (and strength(s) for the SOI and TIOLI models) of the vote outcomes such that the vote specified by the model in question yields the given coalition and its complementary coalition (the set of allowed coalitions is closed under complement due to our binary outcome assumption).  

The set of allowed coalitions for each model depends heavily on the configuration of voter ideal points, and the complexity of the configurations that can occur rapidly increases with the dimension $d$ of the ambient Euclidean space.  We start here with dimension one, where it is quite straightforward to characterize everything completely.  In this setting the voter ideal points $v^{(i)}$ and the outcome locations $p,q$ are all scalars.  Let us assume without loss of generality that $p < q$.
 
\begin{proposition}
A subset $I \subseteq \{1,2,\ldots,n\}$ of voters is an allowed coalition for the one-dimensional EDP model if and only if there exists a point $t\in\mathbb{R}$ with $v^{(i)} < t$ for all $i\in I$ and  $v^{(j)} > t$ for all $j\notin I$, or vice versa.
\end{proposition}

\begin{proof}
This follows immediately from the observation that all voters with ideal point to the left of the outcome midpoint $\frac{p+q}{2}$ will vote for outcome $p$ while all those to the right will vote for $q$. 
\end{proof}

This accords with the multitude of politically split Court cases where a liberal bloc of four justices vote one way, a conservative bloc of four justices votes the opposite way, and the central justice is the ``swing vote'' going back and forth between the two blocks---thus single-handedly determining the outcome---depending on the particular case.  This model also allows for rulings where, say, the two most liberal justices form a minority versus the majority of the remaining seven.  

\begin{proposition}
The allowed coalitions for the one-dimensional SOI model are the same as for the EDP model.
\end{proposition}

\begin{proof}
The outcome spheres in one dimension are just closed intervals, and we always assume in this model that there are no voter ideal points in the intersection of the two outcome spheres, so all voter ideal points in one interval are separated from the other interval either by any point between the two intervals (if they are disjoint) or by any point in the intersection of the two intervals (if they are not disjoint).
\end{proof}

\begin{proposition}
A subset $I \subseteq \{1,2,\ldots,n\}$ of voters is an allowed coalition for the one-dimensional TIOLI model if and only if there exists an interval $[a,b] \subseteq \mathbb{R}$ such that $v^{(i)} \in [a,b]$ precisely when $i\in I$, or precisely when $i \notin I$. 
\end{proposition}

\begin{proof}
This follows immediately from the definition of the preference functions for this model.
\end{proof}

While we see here that the one-dimensional TIOLI model allows votes like the bizarre minority coalition of far left Ginsburg with far right Thomas that occurred in the 7-to-2 case \emph{Alabama Department of Revenue v.\ CSX Transportation, Inc.} (2015)---see Figure \ref{fig:1mds2}---none of our three one-dimensional models allows the scrambled 5-to-4 vote mentioned in the introduction.  To get a better grasp of that vote we need to turn to the second dimension of the Supreme Court.

\begin{figure}
\begin{center}
\includegraphics[trim={0.8in 1.015in 0.4in 0.8in},clip,scale=2.5]{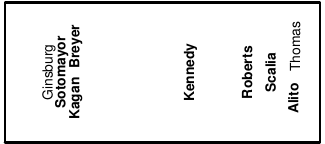}
\end{center}
\caption{A one-dimensional multidimensional scaling layout of judicial ideal points for the 2009--2015 natural court, with the names for the majority in \emph{Alabama Department of Revenue v.\ CSX Transportation, Inc.} (2015) in bold.}
\label{fig:1mds2}
\end{figure}

\subsection{Coalitions in two dimensions (and higher)}

The classification of allowed coalitions for the EDP model has an immediate, and straightforward, extension to all higher dimensions: 

\begin{proposition}\label{prop:EDPhigher}
The allowed EDP coalitions are the sets of voters whose ideal points can be separated from the remaining ideal points by a hyperplane in $\mathbb{R}^{d-1}\subseteq \mathbb{R}^d$.
\end{proposition}

\begin{proof}
Consider the line segment joining the outcome points $\mathbf{p},\mathbf{q}\in\mathbb{R}^d$.  The hyperplane orthogonal to this line segment and passing through its midpoint provides the decision boundary where on one side all voters choose one outcome and on the other side they all choose the other outcome.  Conversely, given a hyperplane separating the two coalitions, we can take as the outcome locations any point on one side of the hyperplane and the reflection of this point across the hyperplane.
\end{proof}

\begin{remark}\label{rem:Voronoi}
Since one of our points of emphasis in this paper is bringing more computational geometry into voting theory, it is worth noting that the preceding proposition has a nice generalization to the EDP model where more than two vote outcomes are considered: the Voronoi diagram \cite{Voronoi} constructed from the configuration of outcome locations in $\mathbb{R}^d$ subdivides Euclidean space into convex polyhedral regions, one such for each outcome point, and the voters in the region corresponding to a particular outcome are precisely those who in the EDP model will vote for that outcome.
\end{remark}

We can now characterize the relationship between the allowable coalitions provided by these three models, in any dimension.

\begin{theorem}\label{thm:coals}
Fix any dimension $d \ge 1$.
\begin{enumerate}
\item A set of voters is an EDP allowed coalition if and only if it is an SOI allowed coalition.  
\item Every EDP/SOI allowed coalition is a TIOLI allowed coalition, but the converse does not always hold.
\end{enumerate}
\end{theorem}

Thus the SOI model sees exactly the same possible voting patters as the EDP model, but the two models arrive at these voting patterns in different ways, so to speak---whereas the TIOLI model is more flexibility and includes every voting pattern provided by the other two models but allows for strictly more voting patterns. 

\begin{proof}
If a coalition is EDP allowed then by Proposition \ref{prop:EDPhigher} there exists a hyperplane separating the ideal points of the coalition and its complementary coalition.  Place spheres of radius $r$ on either side of this hyperplane, both with centers distance $r$ from the hyperplane and such that the line segment between the two centers is perpendicular to the hyperplane (thus both spheres are tangent to the hyperplane and tangent to each other---see Figure \ref{fig:proof}(a)).  As $r \rightarrow \infty$ the two spheres locally (which is to say, near the fixed configuration of voters) approach the hyperplane, and since there are only finitely many voters there exists a finite value of $r$ for which the SOI vote with these two spheres used as the outcome spheres coincides with the EDP vote corresponding to this hyperplane.  By removing one of the two spheres we also get this vote from the TIOLI model.  

\begin{figure}
\begin{center}
(a)~\includegraphics[trim={1.75in 2.7in 1.55in 1.7in},clip,scale=0.65]{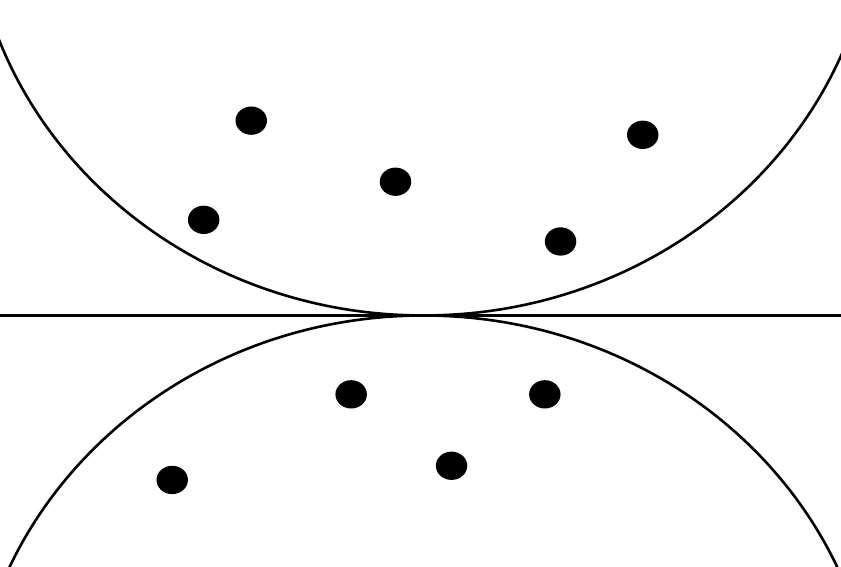}\\
\vspace{0.2in}
(b)\includegraphics[trim={0.0in 2.9in 0.0in 0.0in},clip,scale=0.65]{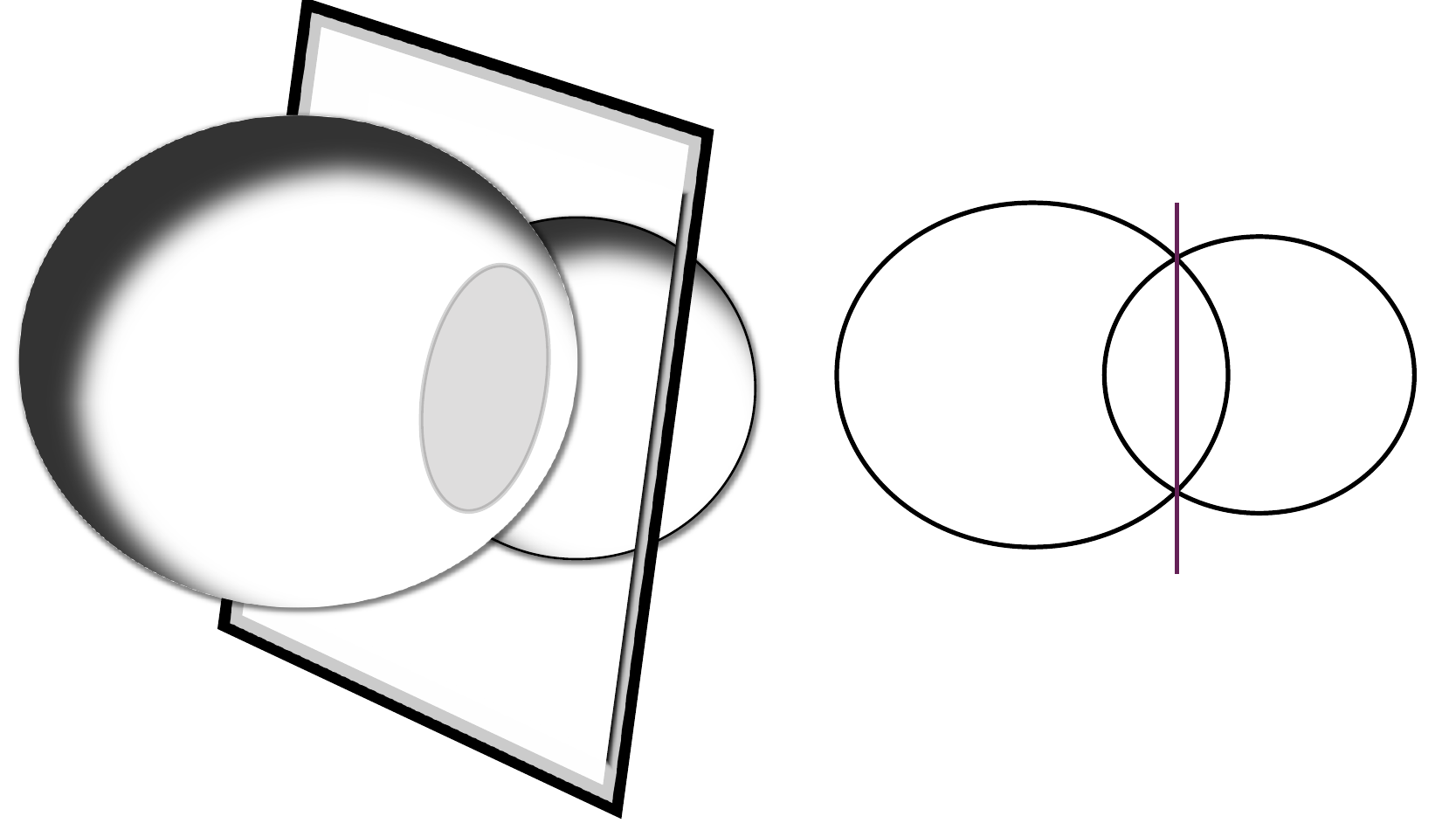}
\end{center}
\caption{Illustrations of the constructions used in the proof of Theorem \ref{thm:coals}.  In (a) we see that the decision boundaries provided by two spheres can be made to approach that of a hyperplane by taking the radii to be sufficiently large.  In (b) we see that, in any dimension, the intersection of the surfaces of two overlapping spheres is the surface of a sphere of one smaller dimension and the hyperplane it spans separates the portions of the solid spheres that lie outside their  intersection.}
\label{fig:proof}
\end{figure}

Thus we have shown that EDP-allowed coalitions are also SOI-allowed and TIOLI-allowed, and we shall see examples in \S \ref{sec:SCexam} of coalitions that are TIOLI-allowed but not EDP- or SOI-allowed, so all that remains is to show that SOI-allowed implies EPD-allowed.  Recall that a SOI-allowed coalition $I \subseteq \{1,\ldots,N\}$ means there is a sphere, call it $B_{\mathbf{p}}$, of radius $s(\mathbf{p})$ centered at $\mathbf{p}$ containing all the ideal points $\mathbf{v}^{(i)}$ for $i\in I$ and none of the ideal points from the complementary coalition $J := \{1,\ldots,N\}\setminus I$, and there is a sphere $B_{\mathbf{q}}$ of radius $s(\mathbf{q})$ centered at $\mathbf{q}$ containing all the $\mathbf{v}^{(j)}$, $j\in J$, and none of the $\mathbf{v}^{(i)}$, $i\in I$.  If these two spheres are disjoint then clearly there exists a hyperplane separating them and hence separating the two coalitions, so let us assume that the spheres overlap, $B_{\mathbf{p}} \cap B_{\mathbf{q}} \ne \varnothing$.  The intersection of the surfaces of these spheres is the surface of a sphere $B'$ of one dimension lower, and the hyperplane $\mathbb{R}^{d-1}$ spanned by $B'$ separates the two coalitions since the two regions $B_{\mathbf{p}} \setminus B_{\mathbf{q}}$ and $B_{\mathbf{q}} \setminus B_{\mathbf{p}}$ lie on opposite sides of this hyperplane (see Figure \ref{fig:proof}(b)) and there are no voter ideal points in the intersection $B_{\mathbf{p}} \cap B_{\mathbf{q}}$ by assumption.
\end{proof}

Our focus for the remainder of the discussion in this section is to explore the coalitions allowed by these two-dimensional models in the context of the Supreme Court.  There are various ways of estimating the ideal points of the justices on the Supreme Court---indeed, this has been a topic of extensive investigation in the empirical legal studies literature \cite{MDS-vs-factor,JudicialCommonSpace,MartinQuinn,Peress-IdealPoint,SegalCover}.  For concreteness and simplicity, in this paper take the following approach (which is also what we used to create Figures \ref{fig:1mds} and \ref{fig:1mds2}).

\subsubsection{Estimating judicial ideal points}
Given a symmetric $N\times N$ matrix of nonnegative numbers, thought of as distances (or approximations thereof) between $N$ objects, and a positive integer $d$, \emph{multidimensional scaling} (MDS) produces $n$ points in Euclidean space $\mathbb{R}^d$ such that their pairwise distances are close to the entries of the input matrix \cite{MDS2,MDS1}.  To use MDS to compute Supreme Court ideal points, scholars typically\footnote{Schubert largely initiated this approach in his seminal 1965 treatise \cite{Schubert}, but Brazill and Grofman in 2002 popularized, modernized, and refined it, while also comparing it (favorably) to related methods like factor analysis \cite{MDS-vs-factor}.} take the $ij$ entry of the $9\times 9$ input matrix to be one minus the fraction of cases during a given natural court that Justice $i$ and Justice $j$ voted ``together''---meaning either both with the majority or both against the majority\footnote{This means the disposition of the decision is what matters rather than the opinion, so for instance if two justices ultimately reach the same conclusion but for different reasons and consequently write separate opinions, they are nonetheless counted as voting together.}---and a widely used way of accessing this information, which we use in the present paper, is the binary ``majority'' variable in the Spaeth Supreme Court Database \cite{SCDB}.  

While the exact coordinates of ideal points vary across the different estimation methods, in one dimension ($d=1$) the order of the nine justices for each natural court is remarkably consistent across the different methods and is generally compatible with expert opinion of the political leanings of the justices \cite[pp.1693--4]{2ndDimension}---so one can view the first dimension in all these methods as quantifying the liberal-to-conservative spectrum.\footnote{A word of caution: MDS is only well-defined up to coordinate reflections, so liberal-to-conservative may show up as left-to-right or right-to-left, but it is usually straightforward to determine which is which and fix the direction.}  The second dimension for MDS ideal point estimation has been studied in \cite{2ndDimension}, where a legal interpretation is proffered: the $y$-axis often roughly tracks with the judicial philosophies of legalism versus pragmatism.  We do not require this interpretation for what follows, but the reader may find it nonetheless enriching to consider and so is encouraged to read \cite{2ndDimension}.

\subsubsection{Examples from the Supreme Court}\label{sec:SCexam}

\begin{figure}
\begin{center}
\includegraphics[trim={0.49in 0.6in 0.39in 0.75in},clip,scale=1.09]{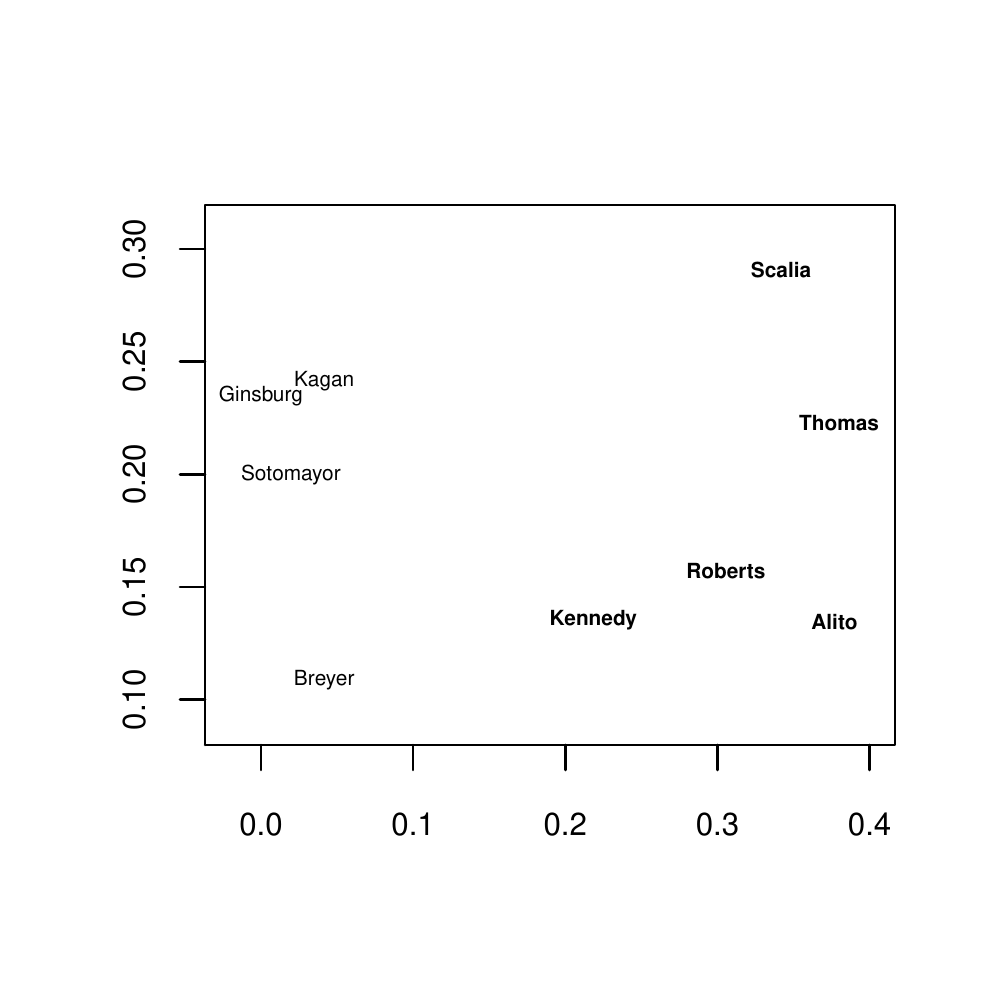}
\includegraphics[trim={0.81in 0.6in 0.39in 0.8in},clip,scale=1.09]{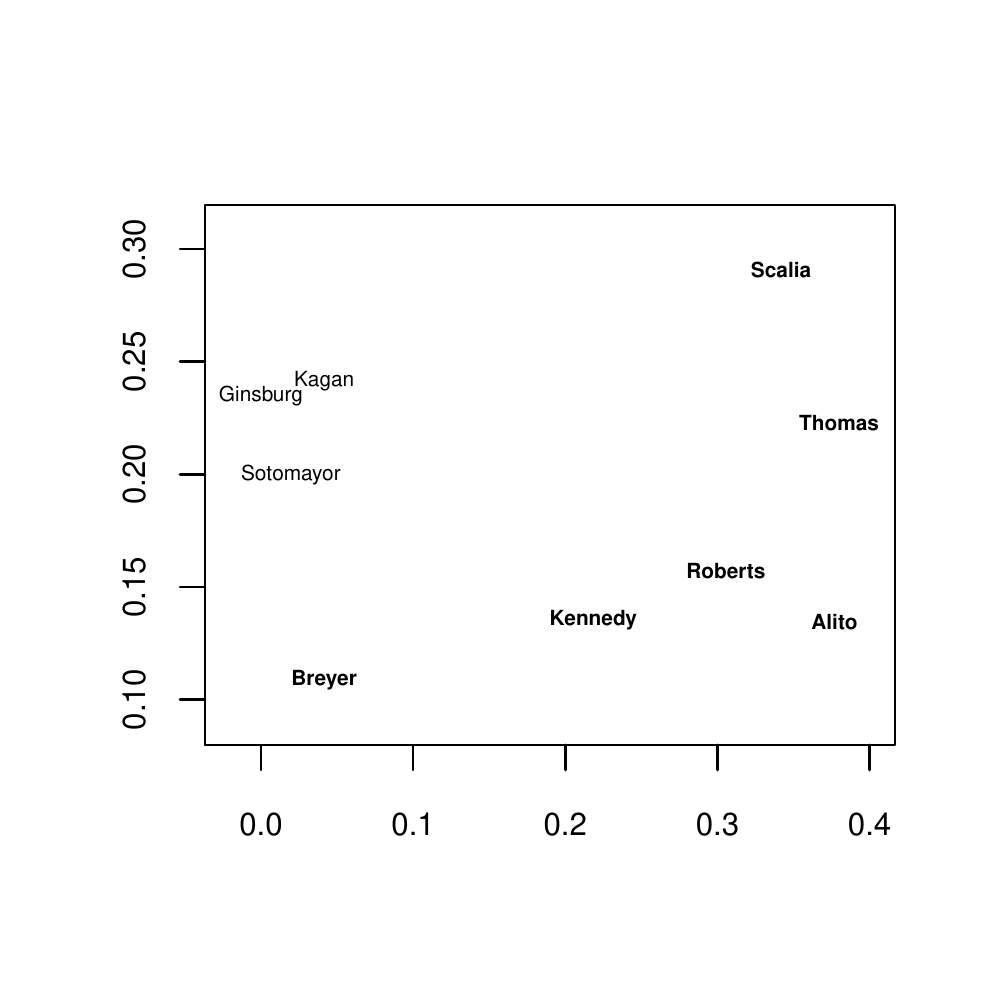}
\includegraphics[trim={0.49in 0.6in 0.39in 0.75in},clip,scale=1.09]{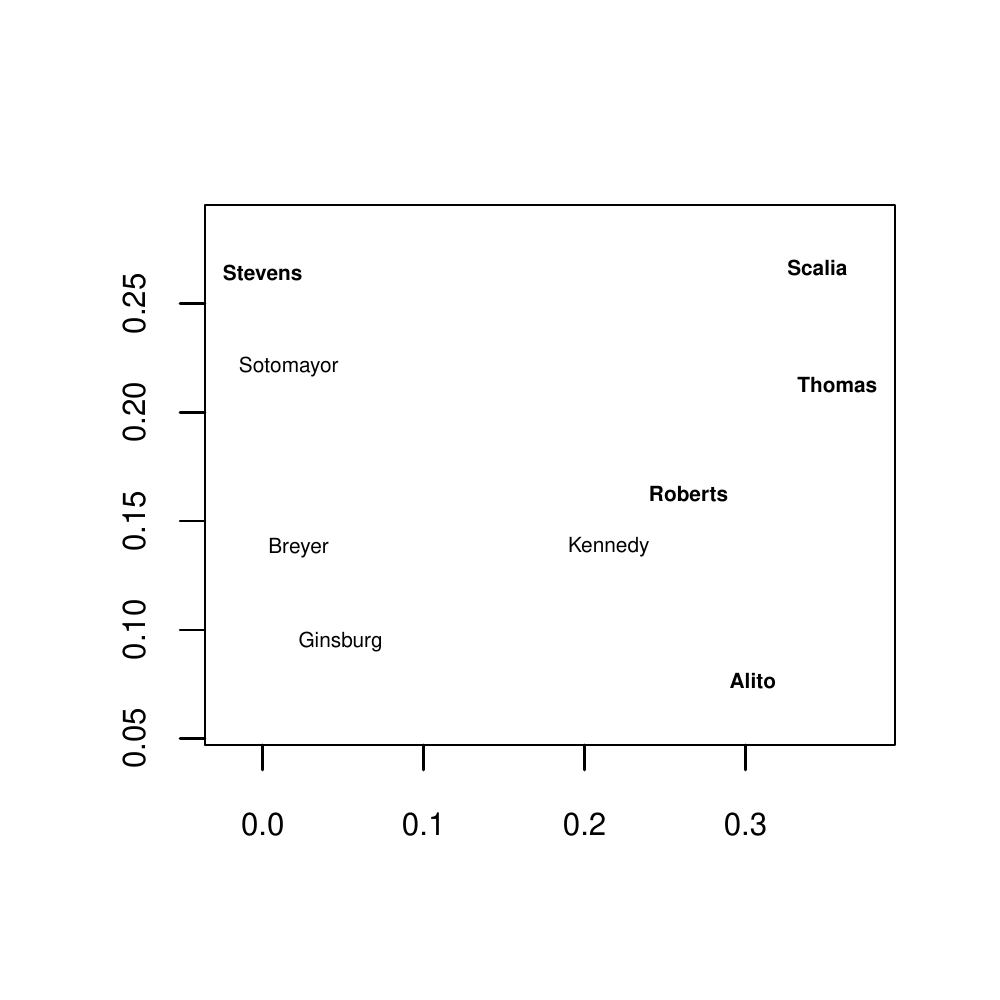}
\includegraphics[trim={0.81in 0.6in 0.39in 0.8in},clip,scale=1.09]{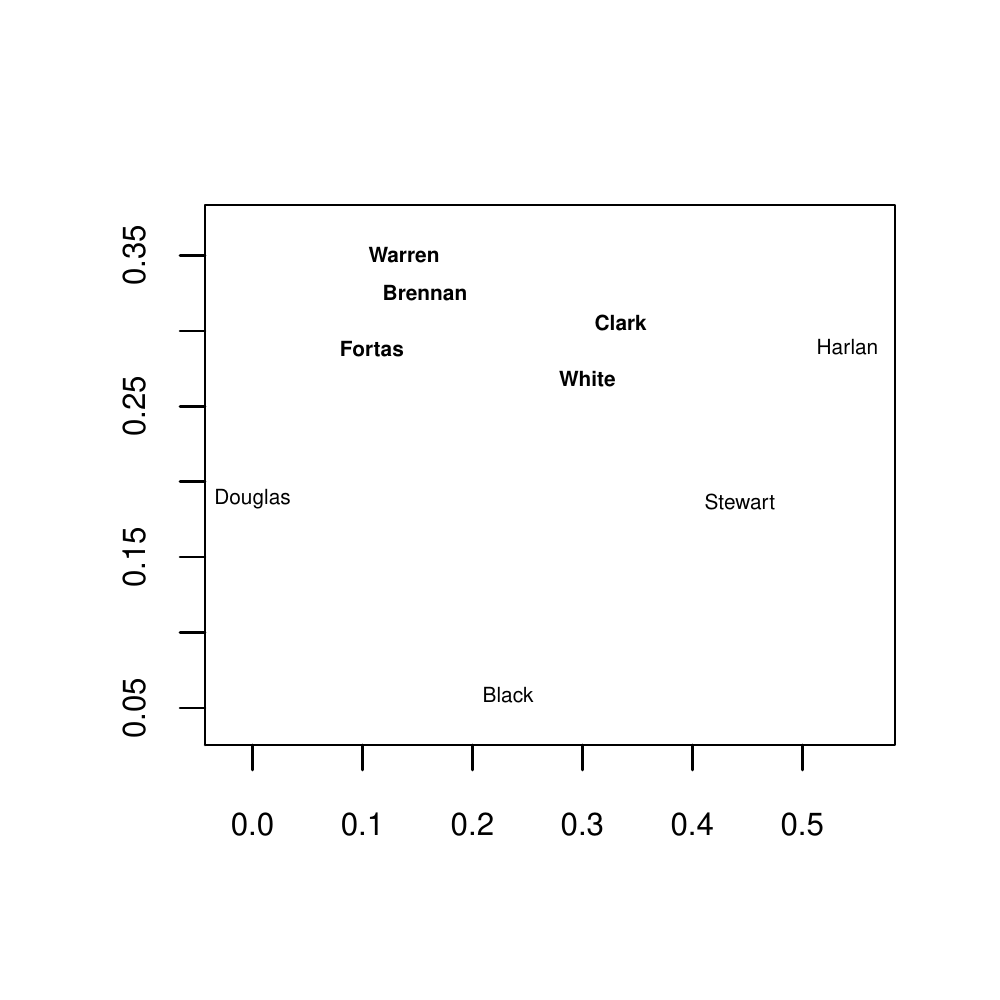}
\end{center}
\caption{The two-dimensional MDS ideal point estimates for the 2009--2015 natural court (top row), the 2009 natural court (bottom left), and the 1965--1966 natural court (bottom right), with bold names indicated the majority coalition for a particular case.  The top-left coalitions are separated by a vertical line and so the case was split along the usual left-right political axis.  The top-right vote---e.g., \emph{J. McIntyre Machinery, Ltd.\ v.\ Nicastro} (2011)---is similar except now Breyer abandoned the liberal bloc and so the dividing line becomes diagonal.  In both the bottom two cases the TIOLI model provides the most robust interpretation, where on the bottom-left---e.g., \emph{New Process Steel v.\ NLRB} (2010)---the minority votes are all in fairly close proximity to a point located around $(0.1, 0.16)$ and in the bottom-right---e.g., \emph{Ginzburg v.\ United States} (1965)---the majority votes are in even closer proximity to each other, with a focal point located around $(0.22, 0.24)$.}
\label{fig:fourcases}
\end{figure}

In Figure \ref{fig:fourcases} we plot the two-dimensional MDS ideal point locations of the nine justices on the 2009--2015 natural court (top row), the 2009 natural court (bottom left), and the 1965--1966 natural court (bottom right); in each of the four plots the names in bold are the majority coalition for at least one case heard during that natural court.  

The top-left plot's minority/majority division is one that occurred frequently: it clearly splits along left/right political lines, and here Kennedy sides with the conservatives (which he did more often than not).  The coalitions for this vote are allowed by all the models we have discussed, including the one-dimensional version of all three models since the one-dimensional MDS ideal point configuration is obtained simply by projecting down to the $x$-axis.  The top-right plot---which occurred, for instance, in \emph{J. McIntyre Machinery, Ltd.\ v.\ Nicastro} (2011)---is more interesting.  Here Breyer has abandoned the liberal bloc, but the two-dimensional layout suggests why: the $y$-coordinate (however it should be interpreted \cite{2ndDimension}) of his ideal point is too far from those of the rest of the liberal bloc.  The result is that the majority/minority coalitions here are no longer divided by a vertical line (the most common conceptualization) nor even a horizontal line (the divisions studied in \cite{2ndDimension}), but they are still separated by a line, albeit a diagonal one, so this pair of coalitions is allowed by the EPD model, and hence by Theorem \ref{thm:coals} by the SOI and TIOLI models as well.

In the lower-left plot---which occurred, for instance, in \emph{New Process Steel v.\ NLRB} (2010)---the two coalitions are almost separated by a diagonal line but not quite: the majority voting Alito and Stevens are too far left for this to be possible.  On the other hand, all four minority justice ideal points are in fairly close proximity to each other and it is easy to picture a circle centered around the point $(0.1, 0.16)$ that encompasses precisely the minority coalition, so the minority/majority coalitions here are allowed by the TIOLI model even though they are not allowed by the EDP or SOI models.  

For the lower-right plot, which depicts a vote that occurred in \emph{Ginzburg v.\ United States} (1965), there is in fact a line that separates the majority from the minority, but just barely.  A slight deviation of the ideal point estimations---and one should recall that there is always error/uncertainty when making such estimations---could easily break this delicate configuration and result in one for which this vote is not allowed by the EDP model.  On the other hand, there is a very clear clustering of the majority ideal points so this pair of coalitions is allowed by the TIOLI model and it is quite robust for that model, in the sense that fairly large perturbations of the ideal points would not prevent this vote division from being allowed.  In fact, the judicial ideal points are almost distributed radially around a point located roughly at $(0.22, 0.24)$ and for this particular case the inner ring of justices near this point formed the majority and the outer ring further away from this point formed the minority.  


\section{Computational geometry}\label{sec:compgeom}

Here we explore some of the algorithmic aspects of working with spatial voting models in two dimensions.  Specifically, we show how to compute the allowed coalitions for a given collection of voter ideal points under each of the EDP, SOI, and TIOLI models, and we discuss the problem---and present a solution for the SOI model---of locating the two outcomes points that could have resulted in a given pair of complementary allowed coalitions.

\subsection{The EDP model}

Recall that, by Proposition \ref{prop:EDPhigher}, the allowed coalitions for the EDP model are the subsets of voters whose ideal points can be separated from those of the complementary set by a hyperplane in $\mathbb{R}^d$.  Such sets of points have been studied in the discrete geometry literature and go by the name ``$k$-set'' (here $k$ refers to the number of points one is trying to separate from the remaining points).  In dimension $d=2$ this problem has been studied from a computational perspective (see, e.g., \cite{k-set}) so for a large number of voters one can implement the algorithms presented there.  For a small number of votes, such as the $N=9$ in Supreme Court voting, there is an easier method that was kindly communicated to us by Michael Burr, which we next describe.  We present this method for $N=9$ and $k=4$, so that one is finding all allowed 5-to-4 votes under EDP, but it is straightforward to extend it to other values of $k$ and other small values of $N$.  We assume here that the voter ideal points are in general linear position, which means no three of them are collinear.\footnote{When working with real data this almost always holds, and if it doesn't then one can simply perturb the points slightly so that it does.}   

At first glance, to find all collections of four points in the plane that can be separated from the remaining five by a line would require testing the infinitely many lines in the plane---but thankfully this can be reduced to a finite problem.  For any line that separates the ideal points into a 5-to-4 division, we can slide that line without changing its slope until it meets exactly one or two of the nine ideal points (by the general linear position assumption).  After doing so, if the line meets only one ideal point then we can rotate it about this point in either direction until it meets exactly two ideal points.  There are then a few possibilities: either both ideal points are in the four-voter coalition, both ideal points are in the five-voter coalition, or one ideal point is in each coalition.  This shows that every EDP allowed 5-to-4 coalition split can be found as follows.  For each of the $\binom{9}{2} = 36$ pairs of ideal points, call them $a$ and $b$, consider the line $L$ they span, and:
\begin{itemize}
\item if there are two and five ideal points on either side of $L$, then the two points together with $a$ and $b$ yields a four-voter coalition;
\item if there are three and four points on either side of $L$, then the four points form a coalition, the three points with $a$ form a coalition, and the three points with $b$ form a coalition.
\end{itemize}
The four-voter coalitions obtained above, together with their complementary five-voter coalitions, form all the allowed EDP coalitions.

It is not really possible to depict in a visually meaningful way (as we shall do below for SOI) the regions where the two vote outcomes could be that yield a given EDP allowed coalition, because the location of one outcome depends on that of the other.  Indeed, for any line that splits the voters into two coalitions, we can take any point on one side of the line as one outcome location and reflect this point across the line to obtain a second outcome location that yields the same dividing line and hence the same pair of coalitions.  An interactive program might be useful: the user could place two points in the plane and the program would draw the corresponding dividing line and then the user could see how the voters are split by this line as the points are dragged around.  In fact, such software already exists---even when more than two vote outcomes are considered (cf. Remark \ref{rem:Voronoi})!  This is because the user is really plotting the Voronoi diagram determined by the outcome locations and there are many programs available, and easily found through a web search, for interactive Voronoi diagrams/tessellations.

\subsection{The SOI model}

For this model we recommend the following discrete approximation method which combines the two computational problems of (a) finding all allowed coalitions and (b) locating the two outcome points.
\begin{enumerate}
\item Choose a large square containing all $N$ of the voter ideal points (with the voters near the center of the square rather than any of its boundary).
\item Choose a grid resolution and do an iterated loop through all grid points in this square.
\item For each grid point, compute the $N$ Euclidean distances from this grid point to the $N$ voter ideal points, put these in increasing order, and record the corresponding order of the voters.  Then for each integer $1 \le k \le \lfloor\frac{N}{2}\rfloor$, label this grid point by the $k$ nearest voter ideal points and also by the $N-k$ nearest voter ideal points.
\end{enumerate}
The allowed SOI coalitions are then the subsets $I\subseteq \{1,\ldots,N\}$ and $J = \{1,\ldots,N\}\setminus I$ for which there exists a grid point labeled by $I$ and a grid point labeled by $J$.  For a fixed pair $I,J$ of complementary allowed coalitions, the possible locations of one vote outcome are all the grid points labeled by $I$ and the possible locations of the other vote outcome are all the grid points labeled $J$.  A word of caution, however: this only works if the surrounding square is sufficiently large and the grid resolution is sufficiently fine; in practice, one can start with small values for both parameters and increase them and see whether the results appear to be stabilizing.

We illustrate this in Figure \ref{fig:SOIplots} where we look at the same two votes as the top row of Figure \ref{fig:fourcases}, but now we additionally plot the majority's possible outcome locations with red dots and the minority's possible outcome locations with blue dots.  The inclusion/exclusion of Breyer's vote in the minority liberal bloc shifts the liberal outcome's region upwards and the conservative outcome's region downwards, as one would expect based on Breyer's location.

\begin{figure}
\begin{center}
\includegraphics[trim={0.49in 0.6in 0.36in 0.75in},clip,scale=1.08]{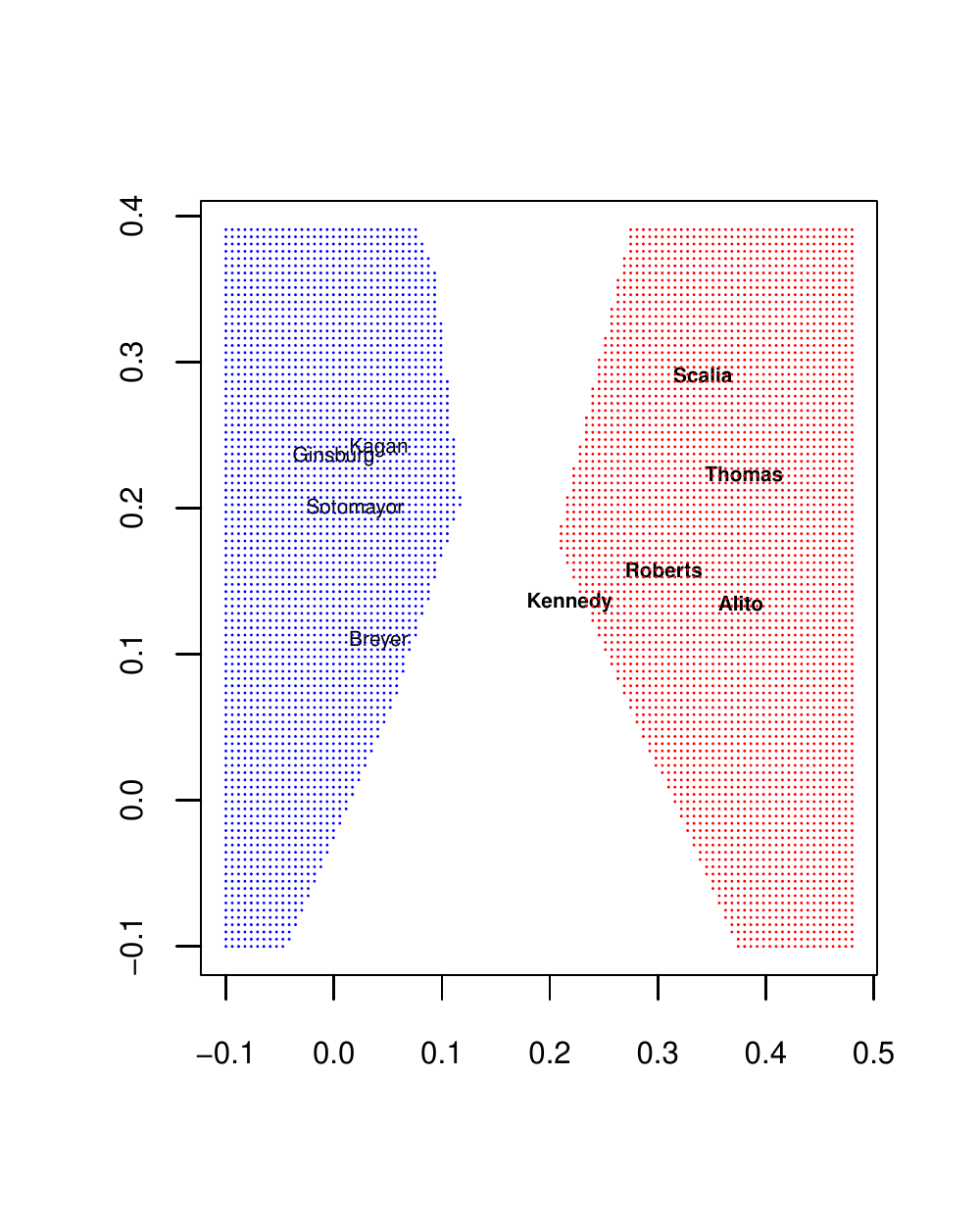}
\includegraphics[trim={0.81in 0.6in 0.36in 0.8in},clip,scale=1.08]{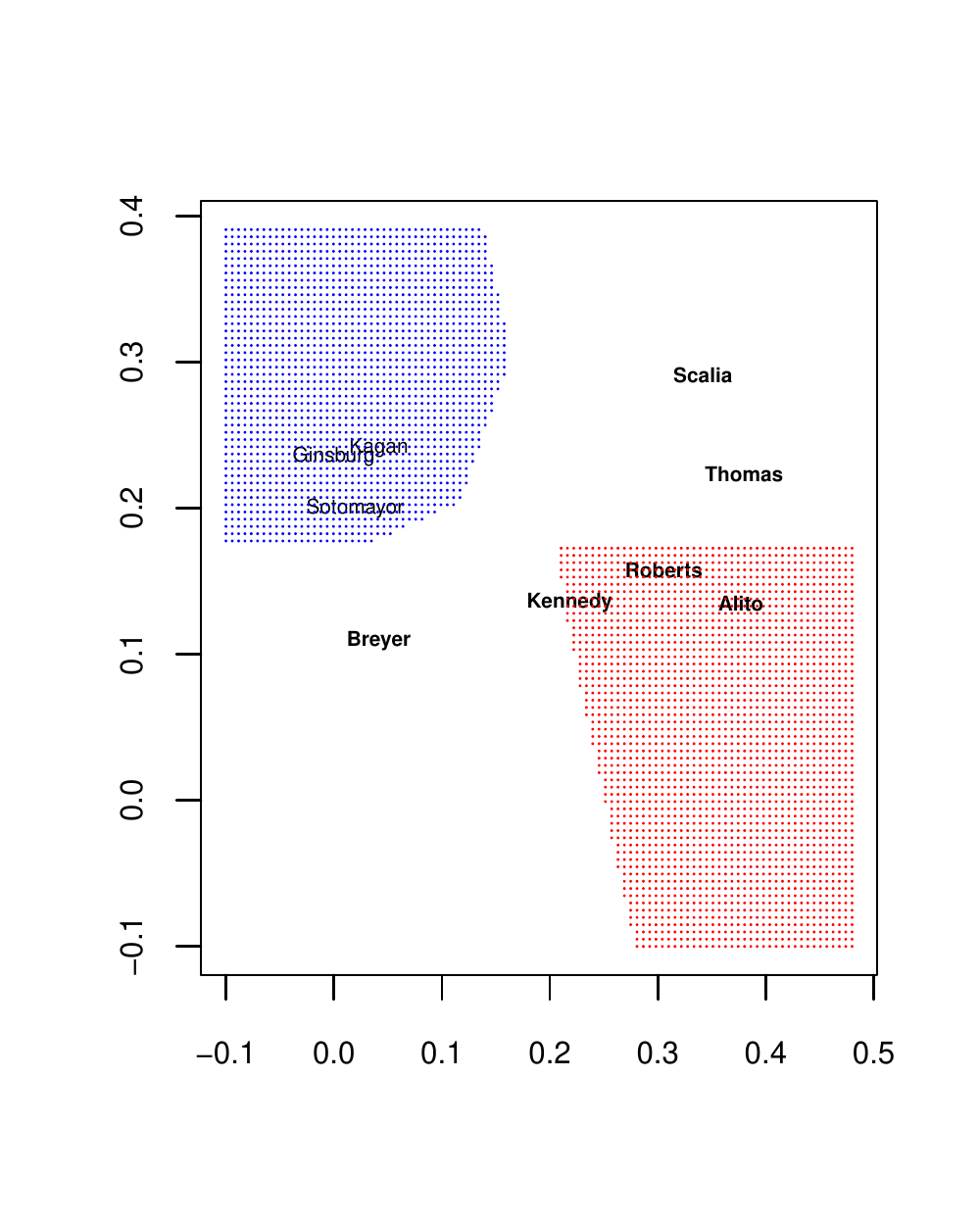}
\end{center}
\caption{Two cases from the 2009--2015 natural court, showing how the possible locations of the liberal outcome (blue dots) and conservative outcome (red dots) shift under the SOI model as Breyer flips from the liberal coalition to the conservative coalition.  These colored regions can be viewed as the \emph{ideological common ground} uniting the justices who voted together in these cases.}
\label{fig:SOIplots}
\end{figure}

\subsection{ The TIOLI model}

We can use essentially the same discrete approximation method described for the SOI model here for the TIOLI model.  Indeed, a subset $I\subseteq \{1,\ldots,N\}$ is a TIOLI allowed coalition if and only if there exists a grid point labeled by $I$, and such grid points are the possible locations for the vote outcome whose outcome sphere contains precisely the $I$ voters.  The location of the other vote outcome is unconstrained and irrelevant as it does not affect the vote under this model.  In Figure \ref{fig:TIOLIplots} we depict the possible locations of the relevant vote outcome for the bottom two cases shown earlier in Figure \ref{fig:fourcases}.

\begin{remark}
Another interesting connection to discrete/computational geometry is that the TIOLI allowed coalitions corresponding to $k$ voter ideal points that lie in an outcome sphere (as opposed to the complementary $N-k$ voters outside the outcome sphere) are in bijection with the cells of the $k$-th order Voronoi diagram defined by the $N$ voter ideal points.  Indeed, given a collection of $N$ points in the plane, the $k$-th order Voronoi diagram is a tessellation of the plane into polyhedral cells for which the set of the nearest $k$ of the $N$ points is constant \cite{HigherVoronoi}.  In particular, another way to find all TIOLI allowed coalitions is to compute the $k$-th order Voronoi diagram for all $1 \le k \le N-1$ and to take the sets of voters whose ideal points correspond to a nonempty cell of at least one of these Voronoi diagrams, as well as the complements to the subsets obtained this way.
\end{remark}

\begin{figure}
\begin{center}
\includegraphics[trim={0.49in 0.6in 0.36in 0.75in},clip,scale=1.08]{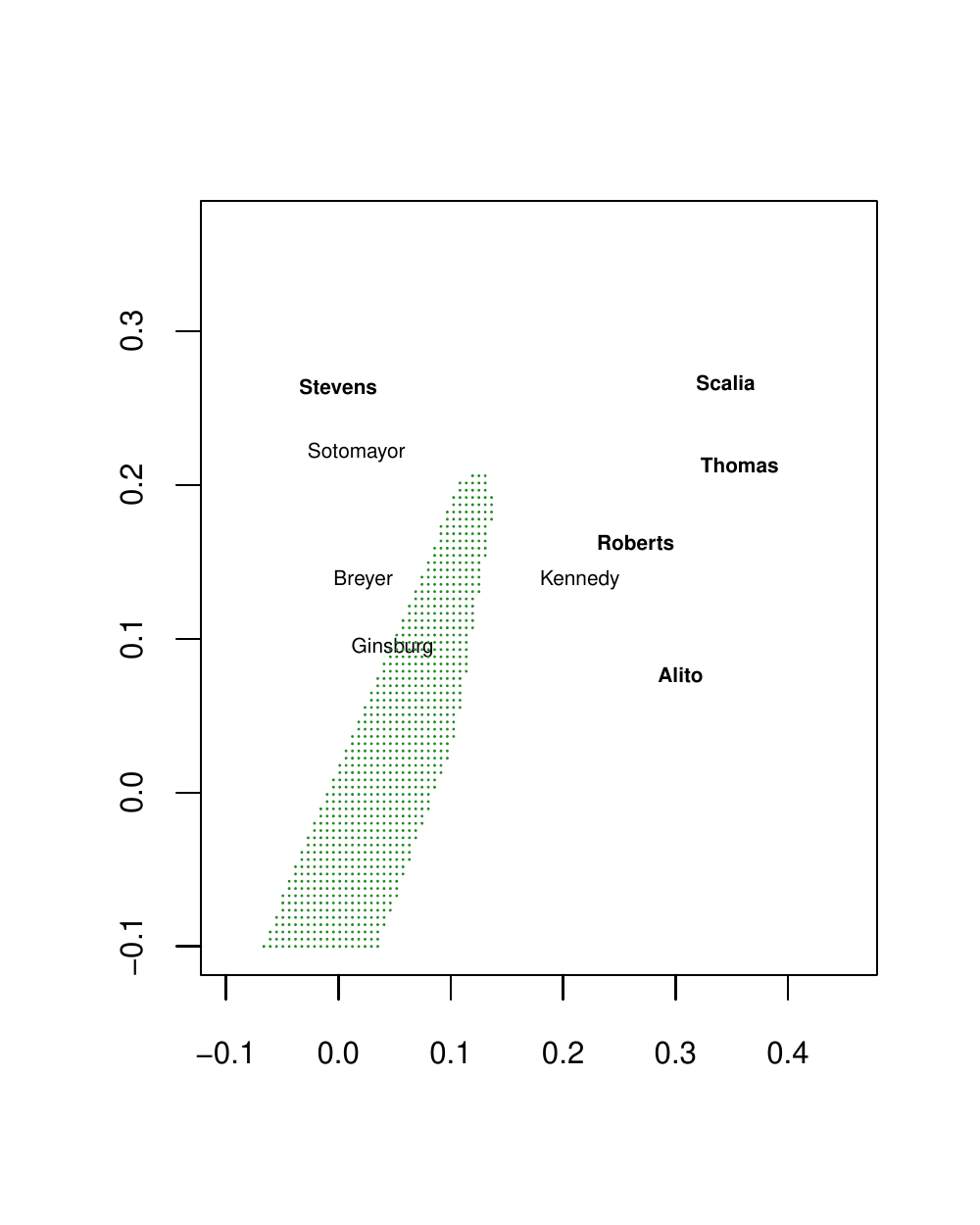}
\includegraphics[trim={0.81in 0.6in 0.36in 0.8in},clip,scale=1.08]{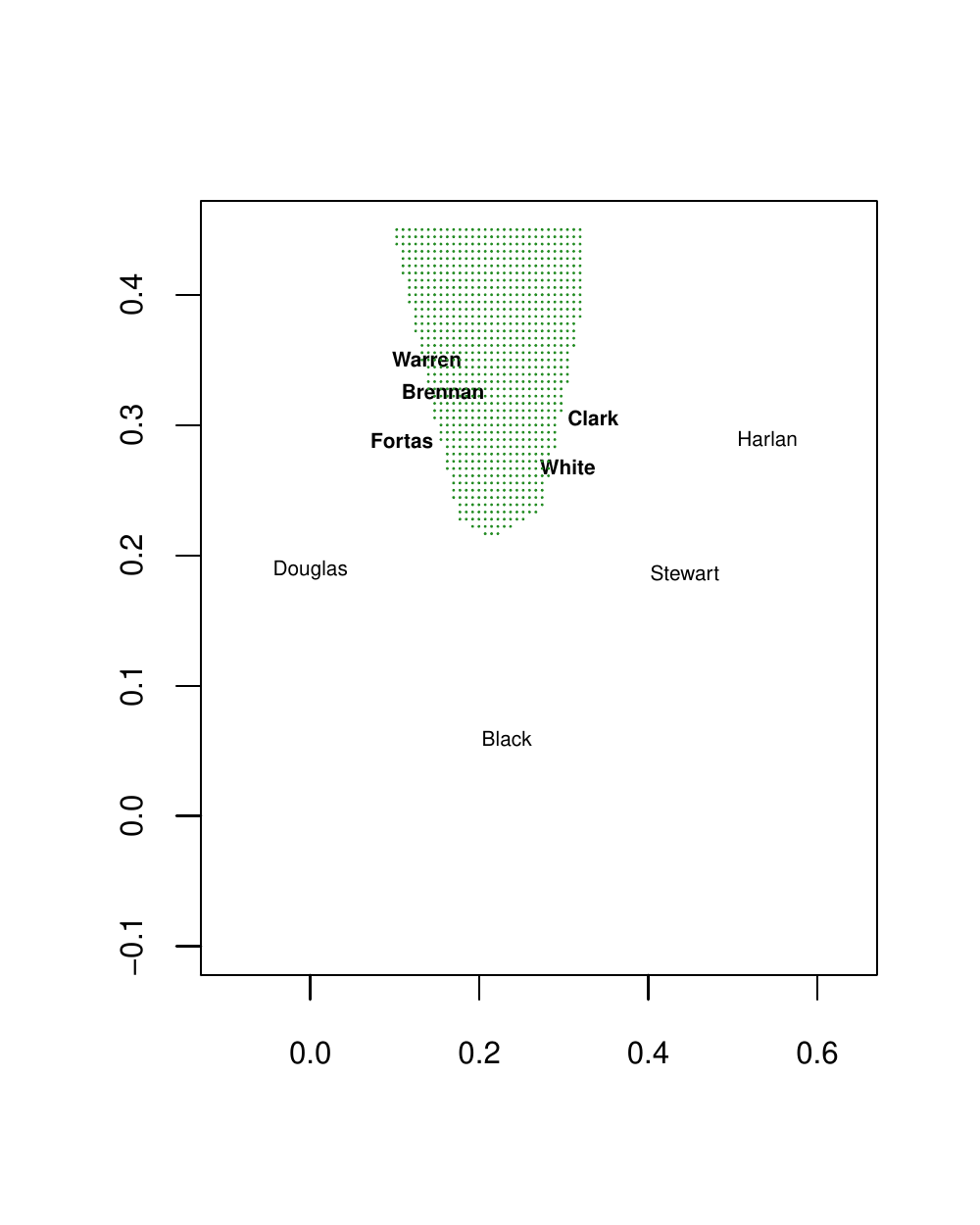}
\end{center}
\caption{The possible locations, indicated with green dots, of the relevant vote outcome provided by the TIOLI model for the 2010 and 1965 cases discussed earlier.  For the left plot these are the locations that can center an outcome sphere containing the four minority votes, and for the right plot they are the locations that can center an outcome sphere containing the five majority votes.  This suggests, for instance, that the minority disposition in \emph{New Process Steel v.\ NLRB} (2010) was aligned with both liberalism and pragmatism, and any justice who is either more conservative \emph{or} more legalistic was drawn to the opposite disposition.}
\label{fig:TIOLIplots}
\end{figure}


\section{Conclusion}\label{sec:conclusion}

The 5-to-4 vote in \emph{Ginzburg v.\ United States} (1965) mentioned in the introduction is quite perplexing from a one-dimension perspective (recall Figure \ref{fig:1mds}), and even when passing to the second dimension of the Supreme Court as in \cite{2ndDimension} the majority is not separated from the minority by a horizontal line.  But Figure \ref{fig:TIOLIplots} shows there is an ideological common ground in the plane uniting the majority justices.  Cases such as this one provide data-oriented evidence that the extra geometric flexibility provided by the TIOLI model introduced in this paper might bear further investigation from a legal perspective.  

It is worth clarifying, however, that the models studied in this paper do not offer a method for predicting Supreme Court outcomes (for that one could see \cite{SC-ML,SupCourtML}), nor do they attempt to explain why the justices voted the way they did in any particular case.  Instead, the main point in these models is to try to understand why certain majority/minority devisions that appear disordered and confounding from a traditional liberal-to-conservative perspective (or even from a pragmatic-to-legalistic perspective as in \cite{2ndDimension}) may perhaps be better viewed through the lens of the outcome spheres introduced in this paper and the consequent notion of ideological common ground revealed through the computational geometry methods discussed above.  That said, there are still plenty of cases resulting in majority/minority divisions that are as perplexing as ever even with all these spatial models---but the complexity of Supreme Court analysis is part of the charm of the subject, for both legal scholars and mathematicians.


\end{document}